\documentclass{article}

\usepackage{PRIMEarxiv}

\usepackage[utf8]{inputenc} 
\usepackage[T1]{fontenc}    
\usepackage{hyperref}       
\usepackage{url}            
\usepackage{booktabs}       
\usepackage{amsfonts}       
\usepackage{nicefrac}       
\usepackage{microtype}      
\usepackage{lipsum}
\usepackage{amsmath}
\usepackage{amsthm}
\usepackage{fancyhdr}       
\usepackage{graphicx}       
\graphicspath{{media/}}     
\usepackage[ruled,vlined]{algorithm2e}
\newtheorem{theorem}{Theorem}
\newtheorem{definition}{Definition}

\title{On the equivalence checking problem for deterministic top-down tree automata}

\author{
  Deng Zhibo\\
  Shenzhen MSU-BIT University\\
  Shenzhen, China\\
  \texttt{dengzb@smbu.edu.cn} \\
   \And
  Vladimir A. Zakharov\\
  Shenzhen MSU-BIT University\\
  Shenzhen, China\\
  \texttt{zakh@cs.msu.ru} \\
}
\begin{document}
\maketitle

\begin{abstract}
We present an efficient algorithm for checking language equivalence of states in top-down deterministic finite tree automata (DFTAs). Unlike string automata, tree automata operate over hierarchical structures, posing unique challenges for algorithmic analysis. Our approach reduces the equivalence checking problem to that of checking the solvability of a system of language-theoretic equations, which specify the behavior of a DFTA. By constructing such a system of equations and systematically manipulating with it through substitution and conflict detection rules, we develop a decision procedure that determines whether two states accept the same tree language. We formally prove the correctness and termination of the algorithm and establish its worst-case time complexity as $O(n^2)$ under the RAM (Random Access Machine) model of computation augmented with pointers.
\end{abstract}


\section{Introduction}

Finite tree automata (FTAs) provide a powerful formalism for processing and analyzing tree-structured data~\cite{tata, tatg}. Unlike finite-state automata (FSAs) that operate on strings, FTAs extend automata theory to hierarchical structures, making them particularly useful in various domains, for example, XML processing~\cite{XML} and program verification~\cite{verif1, verif2}. The ability of FTAs to recognize and manipulate structured data makes them essential tools in both theoretical computer science and practical applications.

FTAs can be classified into two categories based on how they process input trees: top-down and bottom-up tree automata. In a top-down FTA, computation starts at the root of the tree and moves toward the leaves, whereas in a bottom-up FTA, computation proceeds from the leaves to the root. Despite this operational difference, both top-down and bottom-up FTAs are known to have the same expressive power, as they recognize the same class of tree languages~\cite{tata}. However, from an algorithmic point of view they do not have the same behaviour, so the decision problems of them (e.g. determinization, minimization, and equivalence checking) require different approaches. 

The closure of regular tree languages under complementation and intersection easily ensures the decidability of the equivalence problem. As for its complexity, in~\cite{deciding_eq} Helmut Seidl proved that in the general case it is DEXPTIME-complete and PSPACE-complete when restricted to automata accepting finite languages. However, in the same paper~\cite{deciding_eq} it was shown that for any fixed constant $n$, the equivalence of $n$-ambiguous finite tree automata can be decided in polynomial time. Moreover, for weighted tree automata (WTAs), Manssour et al. (2024) proposed holonomic tree automata method~\cite{wta}, which generalize WTAs to computing generating functions for tree series and linking it to differential algebra, providing new tools for automata equivalence and formal power series generation. Yet another way of checking the equivalence of two automata is to compare their minimized versions~\cite{min, mintop}, but this method is obviously far from the best. It should be noted also that most research has focused on bottom-up tree automata, while the algorithmic toolkit for top-down models is still not well developed.

Our work focuses on the equivalence checking problem for top-down deterministic finite tree automata (DFTA), which consists of deciding  whether two states of a given automaton recognize the same tree language. This decision problem has significant theoretical and practical implications. While equivalence checking is well understood for classical finite-state automata on strings~\cite{fas_eq}, its extension to tree automata introduces additional complexity due to the hierarchical nature of tree structures. The computational challenge is posed by the need to compare tree languages defined by recursive state transitions, and hence the need for efficient algorithms that can handle large state spaces and complex transition rules.

The key idea of the equivalence checking algorithm proposed in our work is inspired by some algebraic techniques developed in previous work on FSA models and their extensions~\cite{zakh1, zakh2}. In these papers an effective framework for equivalence checking of automata was introduced; it was shown that it applicable to finite transducers, two-tape automata, biautomata, and push-down automata as well. One of the obvious advantages of this approach is its uniformity: for every automata model it reduces the equivalence checking problem to the solvability of a system of language equations and then applies a variable elimination technique to solve these equations.

We found that DFTAs exhibit precisely the features that enable a straightforward application of this algebraic technique. Firstly, for every state $s$ of a DFTA recursive functional equations of the form $F_q(f(x_1,\dots,x_k))=f(F_{q_1}(x_1),\dots,F_{q_k}(x_k))$ (basic equations) specify all terms accepted by the automaton at its state. Next, the requirement of the equivalence of two states $q_1$ and $q_2$ can be specified by the additional functional equation of the form $F_{q_1}(x)=F_{q_2}(x)$. Thus, the decidability of the equivalence checking problem is reduced to the checking of solvability of the system of functional equations constructed so far. Finally and the most important, the solvability checking of such a system can be achieved by the variable elimination technique. The key point in the application of variable elimination is the capability to transform any equation of the form $f(F_{q'_1}(x_1),\dots,F_{q'_k}(x_k))=f(F_{q''_1}(x_1),\dots,F_{q''_k}(x_k))$ to a series of more simple reduced equations of the form $F_{q'_i}(x)=F_{q''_i}(x)$ and replace all other occurrences of $F_{q'_i}$ with $F_{q''_i}$. As soon as all the reduced functional variables are eliminated from the basic equations, the solvability of the original system is established. If an equation arises that does not allow reduction, this indicates the unsolvability of the original system. As it can be seen from these principles of the variable elimination technique, the number of new equations to be added to the original system does not exceed the number of states of DFTA. If every reduced variable is processed in polynomial time then this yields an equivalence checking algorithm of polynomial time complexity.   

In this paper, we present an efficient algorithm that can check DFTA equivalence in $O(n^2)$ time. Our method exploits the structural properties of DFTAs, in particular their determinism and compositional semantics, and it operates instead of the DFTA transition system with a system of algebraic equations on the terms. By systematically constructing and simplifying these equations, we get a decision procedure that efficiently verifies whether two DFTA states accept the same set of terms. The results obtained in this work contribute to the development of practical and theoretically sound methods for analyzing tree-based computational models.

\section{Preliminaries}

Let $\mathcal{F}$ be a finite set of functional symbols. A ranked alphabet is a couple $(\mathcal{F}, \emph{Arity})$, where \emph{Arity} is a mapping from $\mathcal{F}$ into $N$. The set of symbols of arity $k$ is denoted by $\mathcal{F}^k$, where each element $f \in \mathcal{F}^k$ is referred to as a $k$-ary function symbol. Symbols of arity $0, 1, \dots p$ are respectively called $\textit{constants, unary}, \dots, p\textit{-ary~symbols}$. Let $\mathcal{X}$ be a set of variables, which is disjoint from the function symbols in $\mathcal{F}$. Unlike constants, variables serve as placeholders that can be substituted with terms. The set of \emph{terms} $T(\mathcal{F},\mathcal{X})$ is the smallest set satisfying the following conditions:
\begin{itemize}
    \item All variables in $\mathcal{X}$ and all constant symbols in $\mathcal{F}^0$ are terms, i.e., $\mathcal{X} \cup \mathcal{F}^0 \subseteq T(\mathcal{F},\mathcal{X})$.
		\item If $t_1, \ldots, t_k \in T(\mathcal{F},\mathcal{X})$ and $f \in \mathcal{F}^k$, then  $f(t_1, \ldots, t_k) \in T(\mathcal{F},\mathcal{X})$.
\end{itemize}
This recursive definition ensures that $T(\mathcal{F},\mathcal{X})$ contains all well-formed expressions constructed using symbols from $\mathcal{F}$ and variables from $\mathcal{X}$. If $\mathcal{X} = \emptyset$, meaning that no variables are present, the set of terms is denoted by $T(\mathcal{F})$. A term is said to be \emph{linear} if each variable appears at most once in it. That is, for every variable $x \in \mathcal{X}$, there is no duplication of occurrences within the term. A \emph{context} is a linear term $ C(x_1,\dots,x_k) \in T(\mathcal{F}, \mathcal{X}) $ where the variables $ x_1, \ldots, x_k $ are distinct elements of $\mathcal{X}$. Formally, if a term $C$ contains exactly $k$ distinct variables $x_1, \dots, x_k$, then it is referred to as a \emph{$k$-context}. The set of such contexts is denoted by $\mathcal{C}^k(\mathcal{F})$, and in the special case when $k=1$, the notation simplifies to $\mathcal{C}(\mathcal{F}) := \mathcal{C}^1(\mathcal{F})$. Contexts serve as templates that allow for structured substitution, enabling operations such as term composition and rewriting. If $C \in \mathcal{C}^k(\mathcal{F})$ and $t_1, \dots, t_k$ are terms, then substituting each occurrence of $x_i$ in $C$ with $t_i$ results in a new term, denoted by $C[t_1, \dots, t_k]$. The syntax and semantics of top-down finite tree automata (FTA) are defined formally as follows:

\begin{definition}
A top-down finite tree automata (FTA) is a 4-tuple $\mathcal{A} = (Q, \mathcal{F}, I, \Delta)$, where $Q$ is a finite set of states, $I \subseteq Q$ is a set of initial states, and $\Delta$ is a set of transition rules of the following type:
\begin{equation}
\label{eq1}
(q, f) \rightarrow f(q_1, q_2, \dots, q_k),
\end{equation}
where $k \geq 0$, $f \in \mathcal{F}^k$, and $q, q_1, q_2, \dots, q_k \in Q$. 
\end{definition}
\noindent
$ \Delta({\mathcal{A}},q) $ denotes the set of all transition rules (1) from a state $ q $ of an FTA $\mathcal{A}$. An FTA $\mathcal{A}$ is \emph{deterministic} (DFTA) if there is only one initial state and no two rules have the same left-hand side. By the size of an FTA $\mathcal{A}$ we mean the total number of letters in its transition rules.
  
FTAs operate on terms --- finite trees whose nodes are marked with symbols in $\mathcal{F}$. A transition rule (1) means that whenever an FTA $\mathcal{A}$ at a state $q$ observes a node marked with $f$, the copies of this FTA at the states $q_1, q_2, \dots, q_k$ move to the successors of this node. When $f$ is a constant, a transition rule $(q, f) \rightarrow f $ is called \emph{terminating}: it fires when an FTA reaches a leaf node of an input tree. For two terms of an FTA $\mathcal{A}$: $t$ and $t'$, the \emph{move relation} $t \to_{\mathcal{A}} t'$ is  defined formally as follows:  
$$
t \to_{\mathcal{A}} t' \iff
\begin{cases}
\exists C \in \mathcal{C}(\mathcal{F} \cup Q), \exists t_1, \dots, t_k \in T(\mathcal{F}), \\
\exists q(f) \to (q_1, \dots, q_k) \in \Delta, f \in \mathcal{F}^k,\\
t = C[q(f(t_1, \dots, t_k))], \\
t' = C[f(q_1(t_1), \dots, q_k(t_k))].
\end{cases}
$$
And we defined $\to_{\mathcal{A}}^*$ is the transitive and reflexive closure of $\to_{\mathcal{A}}$. A term $t$ is accepted by $\mathcal{A}$ at a state $q$ if $q(t)\to_{\mathcal{A}}^*t$. A tree language $L_\mathcal{A}(q)$ recognized by $\mathcal{A}$ at a state $q$ is the set of all ground terms accepted by $\mathcal{A}$ at $q$. Two states $q'$ and $q''$ of an FTA $\mathcal{A}$ are called \emph{equivalent} iff $L_{\mathcal{A}(q')} = L_{\mathcal{A}(q'')}$. The equivalence checking problem for FTAs is that of checking, given an FTA $\mathcal {A}$ and a pair of states $q'$ and $q''$, whether these states are equivalent. 

It easy to prove by the induction of the size of terms that the collection of tree languages $L_\mathcal{A}(q)$ accepted at various states $q\in Q$ of $\mathcal{A}$ can be specified as the solution of the following system of equations over the sets of finite terms:
\begin{equation*}
L_{\mathcal{A}}(q) = \bigcup_{\delta \in \Delta(\mathcal{A},q)} \{f(L_\mathcal{A}(q_1), L_\mathcal{A}(q_2), \dots, L_\mathcal{A}(q_k)) \mid \delta: (q, f) \rightarrow f(q_1, q_2, \dots, q_k)\} 
\end{equation*}
A state $q \in Q$ is said to have a \emph{non-empty language} if $L_{\mathcal{A}}(q) \neq \emptyset$, i.e., there exists at least one term accepted by $\mathcal{A}$ starting from state $q$. It was proved that the emptiness problem is decidable for an DFTA, and it is P-complete under logspace reductions~\cite{tata}. The emptiness checking algorithm is presented below.
\begin{algorithm}[H]
\label{alg_empty}
\SetAlgoLined
\KwIn{A top-down DFTA $\mathcal{A} = (Q, \mathcal{F}, I, \Delta)$, a state $q \in Q$}
\KwOut{\textit{True} if $L_{\mathcal{A}}(q) \neq \emptyset$, otherwise \textit{False}}

$\mathit{NonEmpty} \gets \emptyset$\;
$\mathit{DependencyMap}: Q \to \text{set of rules}$ \tcp*[r]{Key: child state, Value: parent rules}
$\mathit{Worklist} \gets$ empty queue\;

\ForEach{rule $\delta: (p, f) \to f()$ where $f \in \mathcal{F}^0$}{
    $\mathit{NonEmpty} \gets \mathit{NonEmpty} \cup \{p\}$\;
    $\mathit{Worklist}.\mathsf{push}(p)$\;
}

\ForEach{rule $\delta: (p, f) \to f(q_1, \ldots, q_k)$ where $k \geq 1$}{
    \ForEach{$q_i \in \{q_1, \ldots, q_k\}$}{
        $\mathit{DependencyMap}[q_i] \gets \mathit{DependencyMap}[q_i] \cup \{\delta\}$\;
    }
}

\While{$\mathit{Worklist}$ is not empty}{
    $q_c \gets \mathit{Worklist}.\mathsf{pop}()$\;
    \ForEach{rule $\delta: (p, f) \to f(q_1, \ldots, q_k)$ in $\mathit{DependencyMap}[q_c]$}{
        \If{$p \notin \mathit{NonEmpty}$ and $\forall q_i \in \{q_1, \ldots, q_k\}, q_i \in \mathit{NonEmpty}$}{
            $\mathit{NonEmpty} \gets \mathit{NonEmpty} \cup \{p\}$\;
            $\mathit{Worklist}.\mathsf{push}(p)$\;
        }
    }
}
\Return $q \in \mathit{NonEmpty}$\;
\caption{Non-Emptiness Checking For Top-Down DFTA}
\end{algorithm}

This algorithm computes the set of states that recognize a non-empty language by iteratively propagating non-empty information. The algorithm is initialised with sets of non-empty states with terminating transition rules that correspond to constants in the input alphabet. $\mathit{DependencyMap}$ is constructed to record each state, and the set of transition rules in which that state appears as a substate. During execution, the algorithm maintains $\mathit{Worklist}$ of newly discovered non-empty states. Whenever a state is removed from the worksheet, the algorithm checks all the transition rules that depend on that state and checks if its parent state becomes non-empty. This process continues until no new non-empty state can be found. The assumption $ L_{\mathcal{A}}(q) \neq \emptyset $ for every state $ q \in Q $ holds iff $ \mathit{NonEmpty} = Q $.

\begin{theorem}
For a top-down DFTA $\mathcal{A}$, algorithm~\ref{alg_empty} always correctly determines whether $L_{\mathcal{A}}(q) \neq \emptyset$. And it runs in time $O(n)$, where $n$ is the total size of $\mathcal{A}$.
\end{theorem}
\begin{proof}
The correctness follows by induction on the construction of $\mathit{NonEmpty}$. A state $p$ is added to $\mathit{NonEmpty}$ if either:
\begin{itemize}
    \item there exists a terminating rule $(p, f) \to f()$, so $L_{\mathcal{A}}(p)$ clearly contains the constant $f$; or
    \item there exists a rule $(p, f) \to f(q_1, \ldots, q_k)$ and all $q_i$ are already in $\mathit{NonEmpty}$, which implies that $p$ recognizes at least the tree $f(t_1, \ldots, t_k)$ where each $t_i$ is accepted from $q_i$.
\end{itemize}
This precisely characterizes the non-emptiness condition. The algorithm terminates because each state and each rule is processed at most once:
\begin{itemize}
    \item Each rule is inserted into $\mathit{DependencyMap}$ exactly once.
    \item Each state is added to $\mathit{NonEmpty}$ and to $\mathit{Worklist}$ at most once.
    \item Each rule in $\mathit{DependencyMap}$ is checked only when one of its child states is popped from $\mathit{Worklist}$.
\end{itemize}
The total number of rule examinations is $O(|\Delta|)$ and each state is pushed and popped at most once, yielding overall time complexity $O(n)$ where $n = |Q| + |\Delta|$.
\end{proof}

\section{Equivalence checking algorithm}
This section presents an equivalence checking algorithm for top-down deterministic tree automata (DFTAs). Our construction follows a similar line to the algebraic method proposed in~\cite{zakh1, zakh2}, which provides a uniform means for designing equivalence checking procedures for a range of finite-state models (e.g. transducers, two-tape automata, and deterministic pushdown systems) by specifying their behavior in terms of language or transduction equations. Our algorithm is specific to top-down DFTAs but it inherits these algebraic insights and leverages the structural determinism of tree automata to maintain a quadratic worst-case time complexity. 

The algorithm consists of two main stages. In the first stage, the algorithm systematically constructs a system of language-theoretic equations that characterizes the equivalence condition $L_\mathcal{A}(q') = L_\mathcal{A}(q'')$ in terms of the structural behavior of the automaton. In the second stage, it applies an iterative variable elimination procedure to this system, based on variable substitution, conflict detection, and equation restoration. This procedure incrementally simplifies the system while preserving its solvability. The algorithm concludes that the two states $q'$ and $q''$ are equivalent if and only if the elimination of variables comes to the end without detecting a conflict, thereby ensuring that the final system admits a solution. Otherwise, a detected conflict confirms their non-equivalence.

Given a deterministic top-down tree automaton $\mathcal{A}$ and two of its states $q'$ and $q''$, the goal is to determine whether $L_{\mathcal{A}}(q') = L_{\mathcal{A}}(q'')$. Our approach divides into two main stages:

\subsection{Stage 1: Constructing the system of equations.}
For each state $q \in \mathcal{A}$, we associate a language variable variable $X_q$ with $q$ to represent the set of trees accepted by the automaton at the state $q$. Then for each transition $\delta$ of the form $(q, f) \to f\bigl(q_1, q_2, \dots, q_k\bigr)$ we denote a term $f(X_{q_1},\dots,X_{q_k})$ as $t_\delta$ and write an equation representing how $X_q$ depends on the variables $X_{q_1}, X_{q_2}, \dots, X_{q_k}$. Then we add the constraint $X_{q'} = X_{q''}$ to the system as our goal requirement that the two states $q'$ and $q''$ should recognize the same tree language. When combining all these equations into a unified equation system $\mathbf{E}_0$ we fully characterize the equivalence condition $L_{\mathcal{A}}(q') = L_{\mathcal{A}}(q'')$: 
\begin{equation*}
{\mathcal E}_0 = \{X_q = \sum\limits_{\delta \in \Delta(q)} t_\delta : q\in Q\} \cup \{X_{q'}=X_{q''}\}. 
\end{equation*}

\subsection{Stage 2: Checking the solvability of the equation system ${\mathcal E}_0$.}
The algorithm iteratively applies the following series of steps to the systems of equations beginning from ${\mathcal E}_0$ built at the Stage 1. Suppose that after the $i$-th iteration, $i\geq 0$, we have a system of equations ${\mathcal E}_i$. which includes equations of the form $X_q = \sum\limits_{\delta \in \Delta} t_\delta$ and $X_p=X_q$.
\begin{enumerate}
\item \emph{Termination detection:} If there are no equations of the form $X_{q'}=X_{q''}$ in ${\mathcal E}_i$ then the algorithm stops and outputs an answer: $q'$ and $q''$ are equivalent states of ${\mathcal A}$. 
\item \emph{Substitution:} Otherwise, for every equation of the form $X_{q'}=X_{q''}$ in ${\mathcal E}_i$ replace all occurrences of the variable $X_{q'}$ in the equations of ${\mathcal E}_i$ with the variable $X_{q''}$ and remove the equation $X_{q'}=X_{q''}$ from ${\mathcal E}_i$.
\item \emph{Conflict detection:} If there are two equations with the same left-hand side (say, $X_q$) such that one of them contains some functional symbol (say, $f$) in its right-hand side, whereas the other does not, then the algorithm stops and outputs an answer: $q'$ and $q''$ are not equivalent states of ${\mathcal A}$.
\item \emph{Restoration:} Otherwise, for every pair of equations with the same left-hand side (say, $X_q$) of the form
\begin{equation}
\label{eq2}
    X_q = \sum f(X_{q_1},\dots,X_{q_k}) 
\end{equation}
\begin{equation}
\label{eq3}
    X_q = \sum f(X_{p_1},\dots,X_{p_k}) 
\end{equation}
and for every functional symbol that appears in their right-hand side (say, $f$) add equations
\begin{equation*}
      X_{q_1} = X_{p_1}, \ \dots \ , X_{q_k} = X_{p_k} 
\end{equation*}
to the system and afterwards remove from the system one of the equations (\ref{eq2}) or (\ref{eq3}). 
\end{enumerate}
Denote by ${\mathcal E}_{i+1}$ the system of equations obtained thus, and this is where the iteration $i$ of the algorithm ends. 	

The final decision to be taken by the equivalence checking algorithm is as follows: either a conflict is found, which means that the states are \emph{nonequivalent}, or the system is successfully resolved, which means that the states are \emph{equivalent}. We provide a formal description of this algorithm:

\begin{algorithm}[H]
\label{alg}
\SetAlgoLined
\KwIn{A top-down DFTA $\mathcal{A} = (Q, \mathcal{F}, I, \Delta)$, states $q', q'' \in Q$ \\}
\KwOut{\textit{True} iff $L_{\mathcal{A}}(q') = L_{\mathcal{A}}(q'')$, otherwise \textit{False}}
$\mathcal{E} \gets \{X_q = \sum f(X_{q_1}, \dots, X_{q_k}) \mid (q, f) \rightarrow f(q_1, \dots, q_k) \in \Delta \}$\;
$\mathcal{E} \gets \mathcal{E} \cup \{X_{q'} = X_{q''} \}$\;
\Repeat{$\mathcal{E}' = \mathcal{E}$}{
    $\mathcal{E}' \gets \mathcal{E}$\;
    \If{no equation of the form $X_{q'}=X_{q''}$ in $\mathcal{E}'$}{
        \Return \textbf{True}
    }
    \ForAll{$X_{q'}=X_{q''} \in \mathcal{E}'$}{
        Substitute $X_{q'}$ with $X_{q''}$ in all equations of $\mathcal{E}'$\;
        Remove $X_{q'} = X_{q''}$ from $\mathcal{E}'$\;
        \If{conflict detected in $\mathcal{E}'$}{
            \Return \textbf{False}
        }
    }
    
    \ForAll{pairs of equations with the same left-hand side}{
        \If{functional symbols match}{
            Add new equalities $X_{q_1} = X_{p_1}, X_{q_2} = X_{p_2}, \dots, X_{q_k} = X_{p_k}$ \;
            Remove redundant equations\;
        }
    }
}
\caption{Equivalence Checking for Top-Down DFTA}
\end{algorithm}

\subsection{Proofs of algorithm~\ref{alg} and complexity analysis}
\begin{theorem}
\label{equiv_algorithm_correctness}
For every top-down DFTA $\mathcal{A}$ and a pair of its states $q'$ and $q''$, the equivalence checking algorithm~\ref{alg} described above always terminates.
\end{theorem}
\begin{proof}
Let $\mathcal{E}_0$ be the initial system of equations constructed from $\mathcal{A}$. By definition, $\mathcal{E}_0$ contains exactly one variable $X_q$ per state $q \in Q$, plus the extra equation $X_{q'} = X_{q''}$ representing the target condition $L_{\mathcal{A}}(q') = L_{\mathcal{A}}(q'')$. Recall that in each iteration of the algorithm, we apply a sequence of steps (Substitution, Conflict Detection, Restoration, and the checks for Termination or Conflict). Define the following characteristics for the current system $\mathcal{E}_i$ at iteration $i$:
$$
(n_1, n_2, n_3),
$$
where 
\begin{itemize}
\item $n_1$ is the number of distinct variables $X_q$ in $\mathcal{E}_i$,
\item $n_2$ is the number of equations in $\mathcal{E}_i$ whose right-hand side contains at least one functional symbol from $\mathcal{F}$ (i.e.\ a ``non-variable'' right-hand side),
\item $n_3$ is the total number of equations in $\mathcal{E}_i$.
\end{itemize}
We order these triples lexicographically, i.e.\ $(n_1, n_2, n_3) < (m_1, m_2, m_3)$ if and only if 
$$
n_1 < m_1 \quad\text{or}\quad \bigl(n_1 = m_1 \text{ and } n_2 < m_2\bigr)\quad\text{or}\quad \bigl(n_1 = m_1, n_2 = m_2 \text{ and } n_3 < m_3\bigr).
$$
In each iteration we perform the following steps in sequence:
\begin{itemize}
\item \emph{Substitution step}: It strictly decreases $n_1$ (we merge two variables into one).
\item \emph{Conflict detection step}: If a conflict is found, the algorithm halts immediately.
\item \emph{Restoration step}: When merging equations with the same left variable, we either unify them (adding new constraints for child variables) or remove duplicates.
\end{itemize}
These operations can decrease $n_2$ or $n_3$ or both. Hence, in each iteration at least one of $n_1$, $n_2$, or $n_3$ is strictly reduced. Because these values cannot go below zero, the process can only proceed for finitely many iterations before it halts (either by detection of conflict or reaching the termination condition). Thus, the algorithm always terminates.
\end{proof}

\begin{theorem}
\label{equiv_algorithm_correctness}
For every top-down DFTA $\mathcal{A}$ and a pair of its states $q'$ and $q''$, the equivalence checking algorithm~\ref{alg} correctly recognizes 
whether $L_{\mathcal{A}}(q') = L_{\mathcal{A}}(q'')$..
\end{theorem}
\begin{proof}
To prove the correctness of the algorithm, we must prove that:
\begin{itemize}
\item[(1)] If the system $\mathcal{E}_0$ is solvable (i.e.\ has a solution over tree languages), then $q'$ and $q''$ are equivalent states;
\item[(2)] The transformations of the system in each iteration preserve the solvability of the system;
\item[(3)] If the algorithm halts by \emph{termination detection} (i.e.\ no more equations of the form $X_p = X_q$), then the resulting system has a solution (so $q',q''$ are equivalent);
\item[(4)] If the algorithm halts by \emph{conflict detection}, then the resulting system has no solution ($q',q''$ are not equivalent).
\end{itemize}

\textbf{(1) \& (2): System $\mathcal{E}_0$ and preservation of solvability.}
By construction, $\mathcal{E}_0$ consists of the equations
$$
\{X_q = \sum_{\delta \in \Delta(q)} t_\delta : q\in Q\} \cup \{X_{q'}=X_{q''}\}. 
$$
where each $t_{\delta}$ is of the form $f(X_{q_1}, \dots, X_{q_k})$ for a transition rule $\delta: (q,f)\to f(q_1,\dots,q_k)$. Observe that there is a natural solution assigning 
$$
X_q \;:=\; L_{\mathcal{A}}(q)\quad\text{for each } q \in Q,
$$
since each $X_q$ must represent the set of all trees accepted in state $q$. Hence $\mathcal{E}_0$ is solvable if and only if $L_{\mathcal{A}}(q') = L_{\mathcal{A}}(q'')$. Furthermore, each local transformation (Substitution, Restoration, removing redundant equations, etc.) does not change whether there exists a solution in which $X_q$'s are interpreted as subsets of trees. In particular:
\begin{itemize}
\item \emph{Substitution} $X_p = X_q$ means identifying the languages represented by $X_p$ and $X_q$. If there was a solution before, we can modify it by letting $X_p$ and $X_q$ interpret the same language; hence solvability is preserved.
\item \emph{Conflict detection} precisely captures the situation where there is no way to assign languages to the same variable to satisfy two incompatible equations.  This indicates no solution exists.
\item \emph{Restoration} merges right-hand sides that share the same functional symbol and adds child-variable equations accordingly, akin to unification of tree structures. If a solution existed before, extending it (or adjusting it) to include new constraints still preserves solvability (unless a conflict is found, in which case solvability fails).
\end{itemize}
Thus, each step preserves the property of ``having a solution'' unless a conflict is explicitly reported.

\textbf{(3) Termination detection implies a valid solution.}
When the algorithm stops because \emph{no equation of the form $X_p = X_q$} remains, the system effectively has each variable $X_q$ associated to a disjunction of terms built from functional symbols (or possibly a pure variable if no merges remain). In such a scenario, the system is consistent with assigning $X_q \subseteq \mathcal{F}^*$ in a way that satisfies all equations. (Concretely, one can read each equation $X_q = \sum_i f(X_{q_1}, \dots)$ as a language-theoretic definition with non-overlapping function symbols.) Since no further merges are pending, the system has no contradictions and can be expanded into a model. Hence there exists a solution, so $q'$ and $q''$ must indeed be equivalent.

\textbf{(4) Conflict detection implies no solution.}
If the algorithm halts in the \emph{conflict detection} step, it has encountered a variable $X_q$ that must simultaneously satisfy two mutually exclusive forms of equations (e.g.\ it is forced to match different function symbols in the same position, or forced to match a function symbol vs.\ an empty (terminating) pattern in the same branch). Such a contradiction makes it impossible to interpret $X_q$ as a single tree-language satisfying both constraints, so the system has no solution. Therefore, $L_{\mathcal{A}}(q') \neq L_{\mathcal{A}}(q'')$.
\end{proof}

We analyze the time complexity of the equivalence checking algorithm~\ref{alg} under the RAM (Random Access Machine) with pointer model. It extends the classical RAM model by incorporating pointer-based memory access. Specifically, it allows the following operations to be performed in constant time ($O(1)$):
\begin{itemize}
    \item Accessing and modifying any element in a data structure via a direct pointer.
    \item Navigating linked data structures (e.g., lists, trees, graphs) by following pointers.
    \item Merging or replacing elements by updating pointers without traversing the entire structure.
\end{itemize}

This computational model reflects the efficiency of practical implementations in modern programming environments, where pointer operations and direct access to memory locations can be executed in constant time, assuming ideal memory allocation and management~\cite{RAM}. Under this model, the cost of key operations in our algorithm—such as variable substitution, equation simplification, conflict detection, and restoration—can be bounded by the size of the current system of equations. 

\begin{theorem}
The equivalence checking algorithm~\ref{alg} always terminates and decides whether $L_\mathcal{A}(q') = L_\mathcal{A}(q'')$ in time $O(n^2)$ where n is the size of $\mathcal{A}$.
\end{theorem}
\begin{proof}
In each iteration, the algorithm either detects a conflict (and stops), or merges two variables (and simplifies the system), or adds constraints and removes duplicate equations. Under the RAM with pointers model, substitution and conflict detection require time proportional to the size of the current system $\mathcal{E}_i$, which is always bounded by the initial size $n$. So each iteration takes $O(n)$ time. Since the number of iterations is at most $O(n)$ (due to the system size parameter) and each iteration takes $O(n)$ time in the pointer model, the total running time is $O(n) \times O(n) = O(n^2)$. Thus, the algorithm can determine the equivalence of $q'$ and $q''$ in $O(n^2)$ time.
\end{proof}

\section{Conclusion}
In this paper we proposed an efficient equivalence checking algorithm for deterministic top-down finite tree automata (DFTAs). The algorithm reduces the equivalence checking problem to the solvability problem for a system of language-theoretic equations. Through a sequence of system transformations, the algorithm iteratively simplifies the system while preserving its solvability. Our algebraic techniques of equivalence checking problem for tree automata contribute to a more shrewd study of equivalence checking problems for finite state models, with potential extensions to top-down tree transducers. Future work may explore optimisation methods that further reduce the actual running complexity and investigate applications in other kinds of models.

\bibliographystyle{unsrt}  
\bibliography{references}

\end{document}